\documentclass[%
 reprint,
 amsmath,amssymb,prl,
 aps,
]{revtex4-2}

\usepackage{graphicx}
\usepackage{dcolumn}
\usepackage{bm}
\usepackage{url}
\usepackage{amsmath}
\usepackage{amsfonts}
\usepackage{mathptmx}
\usepackage{amsthm}
\usepackage{physics}
\usepackage{bm}
\usepackage{xcolor}
\usepackage{hyperref}

\newtheorem{theorem}{Theorem}

\newtheorem{prop}{Proposition}

\theoremstyle{definition}
\newtheorem{definition}{Definition}

\begin{document}

\preprint{APS/123-QED}

\title{Integrability and chaos in the quantum brachistochrone problem.}

\author{S. Malikis}
\email[]{malikis@lorentz.leidenuniv.nl}
\affiliation{Instituut-Lorentz, Universiteit Leiden, Leiden, The Netherlands}%

\author{V. Cheianov}
\affiliation{%
 Instituut-Lorentz, Universiteit Leiden, Leiden, The Netherlands}%

\date{\today}

\begin{abstract}
The quantum brachistochrone problem addresses the fundamental challenge of 
achieving the quantum speed limit in applications aiming to 
realize a given unitary operation in a quantum system. Specifically, 
it looks into optimization of the transformation of quantum states 
through controlled Hamiltonians, which form a small subset in the space 
of the system's observables. Here we introduce a broad family of completely
integrable brachistochrone protocols, which arise from a judicious choice 
of the control Hamiltonian subset. Furthermore, we demonstrate 
how the inherent stability of the completely integrable protocols makes them 
numerically tractable and therefore practicable as opposed to their non-integrable counterparts. 
\end{abstract}

\maketitle

\section{Introduction}
Determining the optimal time required to achieve a unitary operation in a given quantum system holds fundamental significance, particularly in the context of the Quantum Speed Limit \cite{PhysRevLett.103.240501,PhysRevApplied.20.014031}. It also carries important implications in the context of quantum control \cite{Khaneja_2001,Dong_2010,Koch_2022} such as quantum computing \cite{nielsen2006,nielsen_chuang_2010}, shortcuts to adiabaticity \cite{Gu_ry_Odelin_2019} and quantum state preparation\cite{Girolami_2019,Yuan_2023}. The quantum brachistochrone problem belongs to the class of formal problems aiming to facilitate this task. It was first introduced by Carlini et al. \cite{carlini2006time}. It can also incorporate other optimization approaches, like the quantum Zermelo navigation \cite{Russell_2015}. Moreover, it has been used to realize experimentally different optimal quantum qubit gates \cite{PhysRevLett.117.170501,han2020experimental,PhysRevResearch.3.043177}. 

At a conceptual level, the quantum brachistochrone problem stems from the challenge of realising a specific unitary transformation within a quantum system through a process known as 'driving.' This process entails a time-dependent manipulation of the system's Hamiltonian within a constrained parameter space defined by the architecture of quantum hardware. While numerous trajectories exist in this Hamiltonian parameter space that yield the same unitary evolution operator, of primary importance is the one that accomplishes the desired outcome with the smallest effort possible. In essence, the quantum brachistochrone is defined as the global minimum of a quantum cost function, which is the dimensionless product of the quantum computation time and the spectral 
norm of the control Hamiltonian. The mathematical intricacy of the quantum brachistochrone problem arises from the limitation imposed on the space of control Hamiltonians, typically confined to a small subset of all system observables.

The formal statement of the brachistochrone problem consists in specifying the Hilbert space, the choice of 
the subset of allowed Hamiltonians 
and the desired unitary. Given such data one proceeds to 
solving the optimisation problem, which translates into a boundary value problem for a system of non-linear ODEs. Several cases that can be solved analytically have been presented in \cite{Carlini_2011,Carlini_2012,Carlini_2013,yang2022minimumtime}. However, the problem in its generality remains
analytically intractable and, in the case of moderately large systems, numerically hard as the common numerical recipes such as the shooting method \cite{burden2011numerical} fail when the initialization of the search starts away from true solution \cite{bulirsch:02,NoceWrig06}. In an attempt to tackle this difficulty, an alternative numerical approach has been explored in \cite{wang2015quantum}, where the problem was reformulated into searching for geodesics curves in the unitary group. Moreover, in \cite{wang2021quantum}, they re-formulate the problem by exploiting the additional symmetries of the prolem. Successful as this numerical recipe may be, it is naturally limited to gates of small dimensionality.

It is worth noting that the dimension of the system of 
quantum brachistochrone ODEs scales exponentially with the size of the quantum system. As the number of dynamical variables increases and the dynamics becomes more unstable 
solving the problem in full generality eventually becomes 
a formidable task. In fact, computational complexity 
endemic to all chaotic non-linear equations \cite{ott2002chaos} is likely to place the quantum brachistochrone problem 
into the same category of numerically hard problems as \textit{ e.g.} 
the long-term forecast in the Lorenz System \cite{DeterministicNonperiodicFlow}. For this 
reason, it is important to identify special cases, which demonstrate stable dynamics, and which can be scaled to 
large system sizes. A particularly important subset of such cases is formed of completely integrable systems. Not only 
do such systems exhibit numerically stable trajectories, but they also (at least in principle) admit for various reductions and solution by quadratures. Furthermore, by virtue of the phenomenon of KAM-stability \cite{kolmogorov1,moser1,arnold1,arnold2, moser2}, such systems give access to a much broader class of numerically easy non-integrable brachistochrone problems achieved through small deformations of the space of allowed Hamiltonians.

In this paper, we present a class of completely integrable brachistochrone problems. This class, in particular, contains the exactly solved cases discovered in previous literature \cite{Carlini_2011,Carlini_2012,
Carlini_2013,yang2022minimumtime}.
Complete integrability is achieved through the suitable choice of the space of control Hamiltonians, which are regarded as a generating set of a Lie algebra. Complete integrability arises from a relationship between the 
the control Hamiltonian subspace and Cartan decompositions of the associated 
Lie group ($\mathfrak{su}(n)$ in our case). The construction of the integrable cases is inspired by the classical integrable Lie-Poisson Hamiltonian systems \cite{Semenov1981}. After describing the general construction of the class 
of completely integrable quantum brachistochrone problems, we proceed to 
the investigation of the numerical stability of quantum brachistochrone equations in systems with a small Hilbert space. We find that 
the numerical stability of integrable protocols is significantly better 
than of an arbitrary chaotic protocol. Furthermore, the numerical data 
hint that with increasing dimension of the Hilbert space, 
the boundary value problem in a generic non-integrable 
case becomes increasingly hard to solve.

The structure of the paper is as follows. First we give a general overview of the quantum brachistochrone and discuss its general properties. Thereafter we describe the construction of a class of integrable brachistochrone protocols. Then we perform numerical investigation of the stability of integrable and non-integrable protocols showing the qualitative difference. We, furthermore discuss evidence that the solution of the problem becomes increasingly more difficult with increasing dimension of the Hilbert space.

\section{General Overview}
\subsection{The optimisation problem and the cost function}
In this paper we only consider quantum systems having a finite-dimensional computational Hilbert space. In practical hardware problems, such a space is usually associated with the soft (low-energy) part of the spectrum of a physical Hamiltonian. The soft spectrum needs to be separated from the rest of the energy eigenvalues by a sufficiently large energy gap in order to avoid unwanted excitations of the system outside the confines of the computational space. 
All observables discussed below are restrictions of the experimentally accessible physical observables to the computational Hilbert space.

Imagine that we want to realise a unitary operator $\hat U_d\in SU(n)$ on 
an n-dimensional computational Hilbert space as 
an evolution operator generated by a time-dependent 
control Hamiltonian. More precisely, 
consider the evolution operator $U(t)$ satisfying the 
Schr{\"o}dinger equation
\begin{equation}
\label{schrodinger}
    i \partial_t \hat U(t)= \hat H(t) \hat U(t),\qquad \hat U(0)=\mathbb{I}.
\end{equation}
We are looking for a time-dependent Hamiltonian 
$H(t)$ such that at the end of computation $t=t_f$, the  unitary evolution operator will satisfy $\hat U(t_f)=\hat U_d$. Generally, for a given $U_d$ there are 
infinitely many such Hamiltonians. The optimisation problem arises from the task of finding the "shortest" trajectory $\hat H(t),$ which is crudely the one that minimizes the protocol duration $t_f$ (see below for clarification). The problem becomes non-trivial if we assume that $\hat H(t)$ can only be chosen from a subspace of the space of all quantum observables.

Before proceeding to the mathematical specifics we recall some basic facts about the quantum brachistochrone problem. Firstly, we note that trivial multiplication of $H(t)$
by a number $\lambda$ leads to the rescaling of the computation time 
$t_f\mapsto t_f/\lambda.$ For this reason, the computation time itself 
is not a good choice of the optimisation functional. Rather, one introduces a dimensionless cost function $t_f\times \max_t ||\hat H(t)||$ where 
$||\dots ||$ stands for the Frobenius norm of an operator. The 
dimensionless optimisation functional admits for further simplification
owing to a gauge symmetry of the original optimisation problem. 
We note that the Schr\"{o}dinger equation is invariant 
under the redefinition $t \mapsto t'=\varphi(t) $ and $\hat H(t) \mapsto \hat H'(t)=\dot \varphi(t) H(\varphi(t)),$ where $\varphi(t)$
is any monotonically increasing function of time satisfying $\varphi(0)=0.$ 
This implies that one solution $H_*(t)$ to the quantum 
brachistochrone problem one generates a family of equivalent solutions parameterised by the gauge function $\varphi(t).$
In particular, one can easily see that there exists a gauge choice such that $\dot \varphi(t) ||\hat H_*(\varphi(t))|| = 
t_f \max_t ||\hat H_*(t)||,$ whilst $\varphi(t_f)=1.$
With this choice of gauge the cost function can be 
written as 
\begin{equation}
\label{varS}
S=\int_{0}^1 dt ||\hat H(t)||
\end{equation}
and the quantum brachistochrone problem becomes a 
variational problem for the functional $S$ with the 
boundary conditions $\hat U(0)=\mathbb I,$ $ 
\hat U(1) = \hat U_d.$

\subsection{The AB decomposition}
The space of all physical observables constrained to 
the n-dimensional computational Hilbert space coincides with 
the space of all Hermitian endomorphisms of that
space. In a given basis, the non-trivial endomorphisms
are represented by traceless Hermitian $n\times n$ 
matrices, which form a vector space naturally endowed 
with a structure of the Lie algebra $\mathrm{su}(n).$
Let $\hat \gamma_k$ be some basis in the vector 
space of $n\times n$ traceless Hermitian matrices,
 for instance the generalized Gell-Mann 
matrices \cite{Bertlmann_2008},
then $\{\hat e_k\},$ where $\hat e_k = - i \hat \gamma_k,$
will form a basis in the defining 
representation of $\mathfrak{su}(n).$ 
In the following we will assume the trace-form orthonormality condition 
${\rm Tr}(\hat \gamma_i  \hat \gamma_j)= \delta_{ij}.$
In a given basis, the time-dependent Hamiltonian can be represented 
by a trajectory in a $n^2-1$-dimensional Euclidean coordinate 
space 
\begin{equation}
\hat H(t) = \sum_{j=1}^{n^2-1} a_j(t) \hat \gamma_j.
\end{equation}

If one had access to the whole Lie algebra, the 
extremals of the functional \eqref{varS}
would be delivered by the time-independent Hamiltonians 
of the form $\hat H_0=i \log(\hat U_d),$ labelled by the 
choice of the branch of the logarithm function. 
In practice, however, the space of available Hamiltonians 
may be restricted, in which case time-independent solutions to the optimisation problem do not generally exist. In such as situation, the optimisation problem for the functional 
\eqref{varS} may still admit for a solution, albeit with a time-dependent trajectory $a_j(t)$ lying within the permitted component of the operator space. Obviously, in order for 
the optimisation problem to have a solution for any 
$\hat U_d\in \mathrm{SU}(n),$ the permitted subspace of 
the operator space should generate the whole space 
of observables $\mathfrak{su}(n)$ in the Lie-algebraic sense. To formalise the constrained optimisation problem we introduce the following definition:
\begin{definition}
\textit{The AB decomposition}: Let $\mathbb{G}$ 
be the vector space of all observables (that is traceless Hermitian operators) on the computational space and let 
$\mathbb A \subset \mathbb G$ be the subspace 
of physically accessible Hamiltonians, 
that is any $\hat H(t)$ has to satisfy 
$\hat H(t) \in \mathbb A $ at all $t.$ 
Then the following orthogonal decomposition 
will be called the AB decomposition
of $\mathbb G:$
\begin{equation}
\label{ABdef}
    \mathbb{G}=\mathbb{A}\oplus\mathbb{B}.
\end{equation}
Here the $\mathbb A$ and $\mathbb B$ subspaces are 
assumed to be orthogonal relative to the trace form.
\end{definition}

We note, that the AB decomposition of the 
space of observables directly translates into 
a decomposition of the corresponding Lie 
algebra. Let $\mathfrak g $
be the defining representation of 
$\mathfrak{su}(n).$ Then the decomposition 
\eqref{ABdef} induces the following 
decomposition of $\mathfrak g$
\begin{equation}
\label{abdef}
\mathfrak g = \mathfrak a + \mathfrak b,
\end{equation}
where 
$
\mathfrak a = \{-i \hat A|\hat A \in \mathbb A\}$, and $ 
\mathfrak b = \{-i \hat B|\hat B \in \mathbb B\}.$

Given an AB decomposition, we can
choose the orthonormal basis of $\mathbb G$ 
in the form $\{\hat \gamma_i \}=\{\hat A_i\}\cup\{\hat B_j\},$ where $\{\hat A_i\}$ and $\{\hat B_j\}$ are the 
orthonormal bases of $\mathbb A$ and $\mathbb B$ 
respectively. The constraint on the control Hamiltonian 
can then be stated as
\begin{equation}
    \mathrm{Tr}\left(\hat H(t) \hat B_j\right)=0, \qquad j=1, \dots , \dim \mathbb B
    \label{constraint}
\end{equation}
with an explicit solution in the form
\begin{equation}
    \hat H = \sum_i \alpha_i \hat A_i. 
\end{equation} 

\begin{definition}
\label{maxdecoposition}
\textit{Operator controllable AB decomposition}: An AB decomposition is called operator controllable if for every $\hat U_d \in SU(n)$, there exists a 
continuous trajectory
$\hat H: [0, 1]\to \mathbb{A}$ that generates $\hat U_d$ in the sense of the 
boundary value problem \eqref{schrodinger}.
\end{definition}

Obviously, for the AB decomposition to be operator-controllable the Lie algebra $\mathfrak g=\mathfrak{su}(n)$ 
must not have any proper subalgebras containing 
$\mathfrak a.$ In other words, $\mathfrak a$ has to be a generating set of $\mathfrak g$ \cite{1220755}.
It is worth noting that operator controllability, which we focus on here, is a stronger requirement than the (pure) state controllability \cite{Dong_2010}.

To construct $\hat U_d$, we have to perform a constrained 
optimisation of the functional \eqref{varS}. It is 
advantageous to state such an optimisation problem in 
terms of the Lagrangian calculus of variations on 
the Lagrangian manifold $\rm{SU}(n)$. To this 
end, we note that 
\begin{equation}
\hat H(t) = i\partial_t \hat U(t)\hat U^{\dagger}(t)
\end{equation}
is an element of the tangent space, which 
has the meaning of the velocity, and 
that the condition \eqref{constraint} is a set of 
$m=\dim \mathbb B$ non-holonomic constraints, which can be imposed with the help of $m$ Lagrange multipliers. 
The corresponding Lagrange functional takes the local 
form
\begin{equation}
\label{action}
    S=\int_0^{1}dt \left[\sqrt{\mathrm{ Tr}(\partial_t \hat{U} \partial_t \hat{U}^{\dagger})}+i\sum_{j=1}^{\dim \mathbb B} \lambda_{j}Tr(\hat B_j \partial_t \hat{U} \hat{U}^{\dagger})\right].
\end{equation}
The Lagrange functional \eqref{action} can be viewed as 
an action describing constrained motion of a point particle 
on a group manifold.

The Euler-Lagrange equations for the 
functional~\eqref{action} 
(see Appendix for details) consist of equations \eqref{constraint}, due to variation with respect to 
the Lagrange multipliers $\lambda_i,$ and the 
quantum brachistochrone equation 
\begin{equation}
\label{generalbrach}
\frac{d}{dt} (\hat H + \hat D) + i [\hat H , \hat D]=0
\end{equation}
where we introduced the operator 
\begin{equation}\label{Ddef}
\hat D=\sum_i \lambda_i \hat B_i.
\end{equation}
Equation \eqref{generalbrach} was first derived in \cite{carlini2006time,carlini2007time} using a slightly different approach. In a given orthonormal basis $\{\hat \gamma_i\}$ equation \eqref{generalbrach} is equivalent to a system of 
non-linear differential equations for the coordinates $\alpha_i, \lambda_i$ on the tangent bundle of the unitary group \cite{wang2015quantum}
\begin{gather}
    \dot a_i(t)=i \sum_j \lambda_j \text{Tr}(\hat H [\hat A_i, \hat B_j]), \nonumber \\ \dot \lambda_i(t)=i \sum_j \lambda_j \text{Tr}(\hat H [\hat B_i, \hat B_j]). 
    \label{systembrachistochrone}
\end{gather}
The Euler-Lagrange equations are supplemented by the 
initial condition $\hat U(0)=\mathbb{I}$ and the condition that 
$\hat U(1)=\hat U_d.$
One can see that the brachistochrone problem turns into the boundary value problem for a set of non-linear ODEs. The 
form of these equations is completely determined by the choice
of the AB decomposition.

Generally, equations \eqref{systembrachistochrone} do not 
admit for an analytic solution. In order to solve the 
boundary value problem, one typically employs numerical methods, such as the shooting method or gradient descent routines \cite{burden2011numerical,bulirsch:02,NoceWrig06}. However, as the number of dynamical variables $(\alpha_i, \lambda_j)$ increases quadratically with the rank of the group, chaos kicks in making the numerical routines increasingly inefficient. Still, even for large $n$ 
one can identify particularly serendipitous AB decompositions, for which numerical convergence remains very good within either the entire phase space or, at least, a sufficiently large basin of stability. 
A natural class of such good AB decompositions is the one associated with {\it completely integrable} quantum brachistochrone 
equations. It is worth noting, that apart from being well-behaved 
in terms of the Lyapunov stability, 
completely integrable cases admit, at least 
in principle, for solutions in terms of purely 
algebraic equations. Furthermore small non-integrable deformations of such integrable cases will give rise 
to equations with large stability islands. 
These observations prompt a natural task of identification and classification of all completely integrable AB 
decompositions.



In the following sections, we present a
Lie-algebraic construction of a large class of AB decompositions leading to completely integrable brachistochrone equations. However, before delving into these details, we shall conclude the present section with discussion of some general properties
of the brachistochrone problem.

\subsection{Properties of the Quantum Brachistochrone Equations}

First, we derive two useful conservation laws associated 
with equations \eqref{systembrachistochrone}.
By multiplying Eqs.~\eqref{systembrachistochrone} to $a_i$ and $\lambda_i$ respectively, performing summations over all $i$ and utilising the cyclicity of trace, we find
\begin{equation}
\frac{d ||\hat H(t)||^2_F}{dt}=0, \quad \frac{d ||\hat D(t)||^2_F}{dt}=0.
\end{equation}

This implies that the initial conditions, both on $\hat H(t)$ and $\hat D(t)$, determine the cost function of the protocol. Actually this stems from a wider symmetry. The equation \eqref{generalbrach} is essentially
\begin{equation}
\label{lax1}
    \frac{d(\hat H+\hat D)}{dt}=\frac{i}{2}[\hat{H}+\hat{D},\hat{H}-\hat{D}].
\end{equation}
Therefore the operators $\hat H+\hat D, \hat H-\hat D$ form a Lax pair. This means that quantities $F_k=Tr((\hat H+\hat D)^k)$ for $k\in \mathbb{N}$ are conserved, even though not all $F_k$'s are algebraically independent. 

Next, we remark that the Lagrange functional \eqref{action} is invariant under global right shifts $U\mapsto U g,$ where $g\in \rm SU(N).$ By virtue of Noether's theorem, this symmetry results in the conservation
of the generalised angular momentum 
\begin{equation}
\label{conservationlaw}
   \partial_t( \hat U^{\dagger} (\hat H+\hat D) \hat U )=0.
\end{equation}
This equation can help determining the evolution of the optimal $\hat H(t),\hat D(t)$. If one knows the initial state at $t=0$ and the unitary at $t=t_f$ (and not the path from $[0,t_f]$) one can find the final state.
This conservation law is useful in our problem since it creates a relationship between the initial, final configuration of $\hat{H}(t)+\hat{D}(t)$ and the final unitary operator $\hat U_d$. Note that a similar property holds for the operator $\hat U_{-D}$ the operator that is generated through the -D operator: $i \partial\hat U_{-D}=-\hat D \hat U_{-D}$
\begin{equation}
    \hat U_{-D}^{\dagger}(t) \Big(\hat H(t)+\hat D(t)\Big) \hat U_{-D}(t)=  \hat H(0)+\hat D(0)
\end{equation}

 Combining these we can get the commutation relation
\begin{equation}
     [\hat U^{\dagger}_{-D} \hat U,\hat H(0)+\hat D(0)]=0.
\end{equation}
Moreover, by expressing the conservation law \eqref{conservationlaw} and using \eqref{schrodinger}:
\begin{equation}
     i \partial_t \hat{U}(t)  = \hat U(t) (\hat H(0)+\hat D(0)) - \hat D(t) \hat U(t).
\end{equation}
we can rewrite the evolution operator in the following 
form (considering $\hat U(0)= \hat U_{-D}(0)= \mathbb{I} $):
\begin{equation}
\label{multiplyunitary}
   \hat U(t) = \hat{U}_{-D}(t) \exp(- i (\hat H(0)+\hat D(0))t)
\end{equation}
We shall use this equation later in our discussion of a special 
case of completely integrable brachistochrone equation.

Conservation laws and related equations following from 
the general symmetries of the problem are useful, however they are not sufficient to make the problem completely solvable. In the next section we show how to narrow down the set of brachistochrone problems 
to those, which are completely integrable in the 
Arnold–Liouville sense.

\section{Integrable SU(N) brachistochrone equation}

Various definitions of integrability exist. Historically, integrability emerged as the property of equations of motion to admit for an analytic solution, achieved through techniques like separation of variables. Subsequently, a mathematical framework evolved, wherein integrability is construed as the commutativity of a sufficient number of independent Hamiltonian flows within a system's phase space. Both manifestations of integrability engender a stable and predictable behavior in a system. Both perspectives will be used in the present section.

We begin with examining the special case of 
time-independent Lagrange multipliers. The AB decomposition resulting in this integrable system 
was introduced in \cite{yang2022minimumtime}.

\subsection{Complete integrability with time-independent Lagrange multipliers}

Consider a class of trajectories such that Lagrange multipliers $\lambda_i$ are constants of motion, i.e.
\begin{equation}
\label{lambdaconst}
\dot \lambda_i = 0 , \qquad\text{for all}\qquad i = 1, \dots, \dim \mathbb B
\end{equation}
For such trajectories the corresponding differential equations for $\alpha_i$'s in \eqref{systembrachistochrone} can be solved explicitly
\begin{equation}
\label{asol}
    \textbf{a}(t)=\exp(\hat{C}\, t)\textbf{a}(0),\quad \hat{C}_{ab}=\sum_i f_{aib} \lambda_i.
\end{equation}
In this expression $f$ is the structure constant of the Lie algebra of observables, i.e. $[\hat e_i,\hat e_j]=\sum_{k} f_{ijk}\hat e_k,$ the index $i$ runs through 
the basis of the $\mathbb B$ subspace, while the indices $a$ and $b$ run 
through the elements of the subspace $\mathbb A.$
Also we denoted the set of all $a_i$'s in a vector $\textbf{a}(t)=\begin{pmatrix}
\begin{array}{lllll}
    a_1(t) & \cdots & 
     &  a_N(t)
\end{array}
\end{pmatrix}^T$.



The explicit solution \eqref{asol} of the reduced 
system \eqref{systembrachistochrone}, is a significant step forward, 
however, it does not automatically guarantee the existence of an 
algebraic solution for the evolution law of the group element 
given in Eq.~\eqref{schrodinger}. Next we demonstrate 
that in the case of time-independent Lagrange multipliers such an algebraic 
solution does indeed exist.
We note that if Lagrange multipliers are constants of motion then 
the operator $\hat D$ is also time-independent.  This simplifies \eqref{multiplyunitary}:
\begin{equation}
\label{simpmultiplyunitary}
    \hat{U}(t)=\exp(i \hat D(0) t) \exp(-i(\hat H(0)+\hat D(0))t).
\end{equation}
For $t=1$ this becomes a algebraic equation defining the initial values $a_i(0)$ and $\lambda_i$ in terms of the desired operator $\hat U_d.$ In combination with 
Eq.~\eqref{asol} it provides a complete algebraic 
solution for the brachistochrone problem. 
It is worth noting, that despite a tremendous reduction 
in the complexity of the original problem, 
equation~\eqref{simpmultiplyunitary}
remains non-trivial due to its non-linearity and 
multivaluedness of its solutions. 


A comprehensive exploration of the conditions leading to trajectories \eqref{lambdaconst} presents an intriguing task, one that remains far from being exhaustively addressed. Nonetheless, a distinctive case stands out, in which these trajectories emerge directly from the structure of the AB-decomposition, enveloping the entirety of the phase space. One general way to achieve this is 
by choosing the $\mathfrak b$ to be a sub-algebra of $\mathfrak g$. In such a case, equation \eqref{lambdaconst} follows from Eq. \eqref{systembrachistochrone} and the fact that $i [\hat B_i, \hat B_j]\in \mathbb B,$ which is 
orthogonal to $\hat H\in \mathbb A$ with respect to trace inner product. This class of brachistochrone equations, along with Eq. \eqref{simpmultiplyunitary} have been first found in \cite{yang2022minimumtime}.

We conclude the recap of the exactly solvable 
case of constant Lagrange multipliers 
with a discussion of a special case of the AB decomposition having a pseudo-Cartan form. 
In such a case, the algebraic equations 
arising from the boundary value problem admit 
for some further simplifications.
Consider a Lie algebra $\mathfrak{g}.$ Then a decomposition
$\mathfrak{g}=\mathfrak{l}+\mathfrak{p}$ is pseudo-Cartan 
if 
\begin{equation}
\label{cartancomm}
    [\mathfrak{l}, \mathfrak{l}] \subseteq \mathfrak{l}, \quad [\mathfrak{p}, \mathfrak{p}] \subseteq \mathfrak{l}, \quad [\mathfrak{p}, \mathfrak{l}] \subseteq \mathfrak{p}.
\end{equation}
It is easy to confirm that if in the decomposition \eqref{abdef} one chooses $\mathfrak a =\mathfrak{p}$ and $\mathfrak b =\mathfrak{l}$  then a brachistochrone protocol is generated with time-independent Lagrange multipliers. Moreover, when 
$\mathfrak l=\mathfrak{su}(N-1),$ its orthogonal compliment $\mathfrak p$ 
is the smallest generating set in $\mathfrak{su}(N),$
that is the smallest set of controlled Hamiltonians giving access to the computation on the whole unitary group. We note that to the best of our knowledge the pseudo-Cartan decomposition was first employed in the context of optimal quantum control in \cite{PhysRevA.78.032327}. 

For AB decompositions having a pseudo-Cartan form it is possible to 
get some further insights into the structure of 
equation~\eqref{simpmultiplyunitary}.

\begin{prop}
Let a brachistochrone problem with an AB decomposition a pseudo-Cartan decomposition. Given two different $\hat U_d, \hat U'_d$, where $\hat U'_d= \exp(i \hat X) \hat U_d \exp(- i \hat X)$ and $\hat X \in \mathbb B$. The two unitary operators have the same cost function.
\end{prop}
\begin{proof}
    Take Eq. \eqref{simpmultiplyunitary} at $t=1$ for  $\hat U_d$. We assume that its solution is $\hat H_0, \hat D_0$. By multiplying both sides with $\exp(i \hat X),\exp(-i \hat X) $ from left and right, respectively, we get $\hat U'_d$ in the lhs. By considering the defining properties \eqref{cartancomm} we know that
    \begin{equation}
        \hat H'_0=e^{i\hat X} \hat H_0 e^{-i\hat X} ,\quad \hat D'_0=e^{i\hat X} \hat D_0 e^{-i\hat X} 
    \end{equation}
    belong in $\mathbb A, \mathbb B$, respectively. This means that this is the solution for the $\hat U'_d$. Therefore the cost function, i.e. the Frobenius norm of $\hat H_0'$ is the same with $\hat H_0$. 
\end{proof}

\subsection{Liouville-integrable brachistochrone equations}

In the previous section, the integrability of the protocol stemmed directly from the conservation of the Lagrange multipliers. Here we introduce a more general AB decomposition that does not require conservation of $\lambda_i$ and thus generates a wider class integrable brachistochrone protocols. The construction will exploit the Lax representation of the equations of motion for 
the generalised Euler-Arnold top \cite{Semenov1981}. 
To streamline discussion, we will not distinguish between $\mathfrak{su}(n)$ and its defining matrix representation. Furthermore,
will not distinguish between $\mathfrak{su}(n)$ and its dual space, making use of the isomorphism between the two due to the trace-form inner product.

We begin by introducing the time-dependent element 
\mbox{$\hat t = -i(\hat H+\hat D),$} and 
writing its expansion in a given orthonormal 
basis $\mathfrak{su}(n),$ as 
\begin{equation}
\hat t = \sum_{j} x_j (t) \hat e_j
\end{equation}
Here $x_i(t)$ is either a dynamical variable or a Lagrange multiplier, 
depending on which subspace, $\mathbb A$ or $\mathbb B,$ the corresponding $\hat e_j$ lies in. The brachistochrone equation 
can then be written in the form 
\begin{equation}
\label{BEProj}
\frac{d}{dt} \hat t = [\hat t, \mathcal P_B \hat t].
\end{equation}
In this equation $\mathcal P_B \in \mathrm{End}(\mathfrak g)$ is the projector onto the B subspace, which implies 
$
\mathcal P_B \hat t = - i \hat D.
$
Our goal is to demonstrate that given an appropriate choice of the projector $\mathcal P_B,$ 
equation \eqref{BEProj} can be viewed as a certain 
limit of a known completely 
integrable dynamical system. To this end we recall the general construction of the Lax 
representation of the Euler-Arnold top.

Let $\mathfrak g= \mathfrak{l}+\mathfrak{p}$ be a pseudo-Cartan decomposition of $\mathfrak{su}(n)$ \cite{helgason1979differential}. We can then 
write the element $\hat t$ in the form $\hat t  = \hat l + \hat s$ 
where $l$ and $s$ lie in $\mathfrak l$, $\mathfrak p$, respectively
\begin{equation}
    \hat l=\sum_{ \mathfrak l} x_i(t) \hat e_i , \quad \hat s=\sum_{ \mathfrak p} x_i(t) \hat e_i
\end{equation}
Now, let us fix some element $\hat a\in \mathfrak p$ and introduce the following Lax matrix
\begin{equation}
\label{lax2}
    \hat L=\hat a \lambda+ \hat l+ \frac{\hat s}{\lambda}, \quad \lambda \in \mathbb C.
\end{equation}
where $\lambda  \in \mathbb C$ is the spectral 
parameter. This matrix will be used to construct the 
brachistrochrone equation in the Lax form.

In order to obtain the second matrix in the Lax pair, 
we introduce a scalar function 
\begin{equation}
\label{phidef}
 \phi(\hat x ) = \mathrm{Tr} \, \varphi(\hat x)
\end{equation}
where 
\begin{equation}
\label{varphidef}
    \varphi(z)= \sum_{k=0}^K c_k z^k.
\end{equation} 
is a degree $K$ polynomial. 
Note that according to Eq.~\eqref{lax1}, 
$\phi(\hat t)$ is a constant of motion under the
brachistochrone evolution.
For a given set of coefficients $\{c_k\}$, we define 
the element 
\begin{equation}
\label{bdef}
\hat b = \left. \nabla \phi( \hat x)\right 
\vert_{\hat x = \hat a} =\varphi'(\hat a)
\end{equation}
which is the gradient of the function $\phi$ 
at $\hat x= \hat a.$  One can easily see that 
$[\hat a, \hat b]=0$ holds for any choice 
of $\hat a$ and $\{c_k\}.$

Within the subgroup of $\mathfrak l$ we perform a further decomposition $\mathfrak l=\mathfrak l_a+\mathfrak l^\perp$, where $\mathfrak l_a$ is the centraliser of $\hat a$
\begin{equation}
    \mathfrak l_a: \{\hat x\in \mathfrak l,\quad [\hat x, \hat a] =0\},
\end{equation}
which is a subalgebra of $\mathfrak l.$

The structure of the subalgebra $\mathfrak l_a$ depends on the order of $\mathfrak{su}(n),$ the choice of the pseudo-Cartan decomposition and the choice of the element $\hat a$. 
We now define the function \mbox{$\phi_a : \mathfrak l_a \to \mathbb R, $ }
$$\phi_a(\hat x)=\phi(\hat a+ \hat x),\quad \hat x \in \mathfrak l_a$$ and introduce its Hessian 
tensor at the point $l=0$:
\begin{equation}
    \phi_a'':=\sum_{i,j\in \mathfrak l_a} \hat e_i \otimes \hat e_j \frac{\partial^2 \phi_a}{\partial x_i \partial x_j}\Big|_{\hat x=0}
\end{equation}
With the help of the Hessian, we define a linear map $\hat \omega$
from $\hat l $ to itself by
\begin{equation}
    \hat \omega(\hat l) =
    \begin{cases}
      \phi''_a \hat l & \text{if $\hat l \in \mathfrak{l}_a$} \\
      \mathrm{ad}_b (\mathrm{ad}_a)^{-1} \hat l & \text{if $\hat l\in \mathfrak{l}^{\perp}$}
    \end{cases}.
    \label{omegadef}
\end{equation}

With these ingredients we complete the Lax pair with the 
matrix $\hat M:$
\begin{equation}
\label{laxpair2}
    \hat M= \hat b \lambda+\hat \omega(\hat l), \quad \forall \lambda \in \mathbb C.
\end{equation}
The following Lax equation 
    \begin{equation}
    \label{LaxEq}
        \frac{d\hat L}{dt}=[\hat{L},\hat M]
    \end{equation}
describes a class of completely integrable 
systems known as generalised Euler-Arnold tops~\cite{Semenov1981}.

We are ready to provide the following Theorem.

\begin{theorem}
    For a given element $\hat p \in \mathfrak p$, define $\hat a=\epsilon \hat p$, where $\epsilon\in \mathbb R$. 
    Let 
    \begin{equation}
    \mathfrak l_a=\mathfrak l_a^{(A)}+\mathfrak l_a^{(B)}
    \end{equation} 
    be some, possibly trivial, decomposition of the centraliser of $\hat a$ into a direct sum of semi-simple subalgebras of $\mathfrak g.$
    Then with an appropriate choice of the polynomial $\varphi(z),$
    Eq.~\eqref{varphidef},
    the $\varepsilon \to 0$ limit of the Lax equation \eqref{LaxEq} 
    coincides with the brachistochrone equation \eqref{BEProj} where 
    the AB decomposition takes the form 
    \begin{equation}
    \label{pla}
        \mathfrak a =\mathfrak p + \mathfrak l_a^{(A)}, \quad \mathfrak b= \mathfrak l^{\perp}+\mathfrak l_a^{(B)}
    \end{equation}
\end{theorem}

\begin{proof}
    Let $\mathrm{Spec}(\hat a)=\{ a_1, a_2, \dots \},$ be the eigenvalue spectrum of $\hat a$ in the defining 
    representation of $\mathfrak g= \mathfrak{su}(n)$ on the complex vector space $v=\mathbb C^n.$ To each eigenvalue 
    $a_i$ there corresponds an eigenspace $v_i\in v$ and a simple
    subalgebra $\mathfrak l_a^{(i)} \subseteq \mathfrak l_a,$ which is the largest subalgebra of $\mathfrak l_a$ preserving $v_i$ and acting 
    trivially on its orthogonal complement $v_i^\perp.$ The 
    centraliser of $\hat a$ has the semi-simple decomposition 
    \begin{equation}
    \mathfrak l_a = \mathfrak l_a^{(1)}+\dots+\mathfrak l_a^{(Q)}
    \end{equation}
    where $Q\le n.$ Without loss of generality, we may assume a numbering such that 
    \begin{equation}
    \mathfrak l_a^{(A)} =
    \mathfrak l_a^{(1)}+\dots + \mathfrak l_a^{(q)}, 
    \qquad
    \mathfrak l_a^{(B)} = \mathfrak l_a^{(q+1)}+\dots + \mathfrak l_a^{(Q)}
    \end{equation}

    Consider now the following choice of the polynomial $\varphi:$ 
    \begin{equation}
    \label{choice}
    \varphi(z) =\frac{z^2}{2}+\frac{1}{2}\psi(z) \prod_{i=1}^Q (z-a_i)^2,
    \end{equation}
    where 
    \begin{equation}
    \psi(z) = -\sum_{k=q+1}^Q
    \prod_{\substack{s=1\\s\neq k} }^Q 
    \frac{z-a_s}{(a_k-a_s)^3}.
    \end{equation}
    One can easily see that with this choice of $\varphi,$ the eigenvalues of $\hat b$ coincide with the eigenvalues of $\hat a$ therefore  $\hat b = \hat a.$ It follows immediately that the 
    restriction of $\hat\omega(\hat l),$ as defined in Eq.~\eqref{omegadef}, to the space $\mathfrak l^\perp$ acts as the identity.

    For the restriction of $\hat\omega(\hat l)$
    to the subspace $\mathfrak l_a$ we have 
    \begin{equation}
    \phi_a''= \sum_{i=1}^Q \varphi''(a_i) \mathcal P^{(i)}
    \end{equation}
    where $\mathcal P^{(i)}\in \mathrm{End}(\mathfrak g)$ is the orthogonal projector onto the subspace $\mathfrak l_a^{(i)}.$
    A straightforward calculation shows that for $\varphi$ given in Eq.~\eqref{choice} one has $\varphi''(a_i) = 0$
    for $i=1, \dots, q$ and $\varphi''(a_i)=1$ for $q<i\le Q.$ 
    Therefore, with $\varphi(z)$ given by Eq.~\eqref{choice} we 
    have
    \begin{equation}
    \label{omegaproj}
    \hat \omega = \mathcal P^\perp + \mathcal P^{(B)}=\mathcal P_B
    \end{equation}
    where $\mathcal P^{(B)}=\mathcal P^{(q+1)}+ \dots 
    +\mathcal P^{(Q)}.$
    
    Finally, substituting the Lax matrices into the Lax equation and \eqref{LaxEq} and gathering coefficients at 
    different powers of $\lambda$, we obtain the following set 
    of equations
    \begin{gather}
    [\hat a,\hat b]=0, \quad [\hat a,\hat \omega (\hat l)]=[\hat b,\hat l], \nonumber \\
    \dot {\hat l} = [\hat l,\hat \omega(\hat  l)]+[\hat s,\hat b], \quad \dot {\hat s} =[\hat s,\hat \omega(\hat l)].
    \end{gather}
    The first pair of equations is satisfied automatically thanks 
    to the definitions \eqref{bdef}, \eqref{omegadef}.
    In the second pair of equations we use the the definition of 
    $\hat t = \hat l+\hat s,$ the fact that $\hat b = \hat a= \epsilon \hat p,$ and the specific form of $\hat \omega,$ Eq.~\eqref{omegaproj} to obtain 
    \begin{equation}
    \frac{d}{dt} \hat t=[\hat t, \mathcal P_B \hat t ]+\epsilon [\hat s, \hat p].
    \end{equation}
    Taking the $\varepsilon \to 0 $ limit we finally obtain the brachistochrone equation in the form \eqref{BEProj}, 
    with the AB decomposition given by \eqref{pla}.
\end{proof}

\subsection{Cases admitting for reduction to a linear system}

There are two special types of the decomposition \eqref{pla},
which admit for a straightforward reduction to a linear 
system. 

\textit{Type I.} In AB decompositions of this Type, the $\mathfrak b$ subspace is a subalgebra of $\mathfrak g.$ This is achieved by
choosing $\mathfrak l_a^{(A)}=0$ and $\mathfrak l_a^{(B)}=\mathfrak l_a$  in \eqref{pla}. This case has been analysed 
in literature~\cite{yang2022minimumtime} and we discuss it here 
in the section on time-independent Lagrange
multipliers.

\textit{Type II}. By choosing $\mathfrak l_a^{(A)}=\mathfrak l_a$ and $\mathfrak l_a^{(B)}=0$ one gets $\mathfrak b = \mathfrak l^\perp.$
Clearly, in this case the $\mathfrak b$ subspace is not a subalgebra of 
$\mathfrak g.$ However, the system 
\eqref{generalbrach} can still be reduced to a linear one.
\begin{prop}
\label{type2prop}
  Consider the AB decomposition such that   \mbox{$\mathfrak b = \mathfrak l^\perp.$} Then the dynamical variables, which are the projections of $\hat t$ onto the $\mathfrak l_a$ subspace are integrals of motion. Furthermore, the system of ODEs \eqref{generalbrach} reduces to a linear system with time-dependent coefficients.   
\end{prop}

\begin{proof}
    If $\mathfrak l_a^{(B)}=0,$ then 
    $\mathcal P_B \hat t = \hat l^\perp$ is the orthogonal projection 
    of $\hat t$ onto $\mathfrak l^\perp.$
    In this case, the orthogonal projection $\hat l_a$ of
    $\hat t=\hat s+\hat l$ onto the subspace $\mathfrak l_{a}$ is a constant of motion. To see that, consider the commutator 
    $[\hat s+\hat l,{\hat l}^{\perp}]$ on the right-hand side of \eqref{BEProj}. The only potentially non-trivial projection
    of this commutator on the $\mathfrak l_a$ subspace is 
    due to the element 
    $\hat x=[\hat l_a, \hat l^\perp].$ However, 
    this element lies in $\mathfrak l^\perp$ because for any 
    $\hat l_a'\in \mathfrak l_a$ one has \mbox{
    $\text{Tr}(\hat l_a' \hat x) =\text{Tr}(\hat l^\perp [\hat l_a', \hat l_a])=0,$} where we have used the cyclicity of trace and the 
    fact that $\mathfrak l_a$ is closed under the Lie bracket. 
    
    We further note that the system of equations for the dynamical variable $\hat l^\perp$ takes the form 
    \begin{equation}
    \frac{d}{dt} \hat l^\perp = [\hat l_a, \hat l^\perp].
    \end{equation}
    Since $\hat l_a$ are integrals of motion, this is a linear 
    system, which has an explicit solution in the form 
    \begin{equation}
    \hat l^\perp(t)= e^{\hat l_a t} \hat l^\perp(0) e^{-\hat l_a t}.
    \label{sperpsol}
    \end{equation}

The remaining equation for the $\hat s$ component of $\hat t$
takes the form 
\begin{equation}
\frac{d}{dt} \hat s = [\hat s, \hat l^\perp(t)],
\end{equation}
where $\hat l^\perp (t)$ is given by Eq.~\eqref{sperpsol}.
This is a linear equation with variable coefficients, which 
is formally solved in the form of a time-ordered exponential.
\end{proof}
Protocols derived based on Proposition \ref{type2prop} fall under the type II protocols.  An example of such a construction is presented in the Appendix for the  $\mathfrak{su}(3)$ algebra.

The AB decompositions we presented rely on the Cartan decomposition of $\mathfrak{su}(n)$ algebras. In the next section we present examples for $n=3,4$, but it can be extended for arbitrary $n.$ 
A question of practical importance, which is worth mentioning here, 
is how the integrable protocols 
could be realised with local operators in multi-qubit systems. This
brings up the issue of the construction of pseudo-Cartan decompositions 
of $\mathfrak{su}(2^n)$ with maximally local gates - a mathematical
problem, where few results are known today \cite{khaneja2001cartan}.

\section{Numerical stability}
\subsection{The stability of equations of motion}
We embarked on the search for integrable brachistochrone equations due to their inherent stability. A well-known issue in general non-linear problems is the instability of trajectories under slightly varied initial conditions. Specifically, if we solve equations with initial conditions $\textbf{x}'(0)=\textbf{x}(0)+\textbf{d}$ with $|\textbf{d}|\ll|\textbf{x}(0)|$ the respective variation for the solutions versus time $DH(t)=||\textbf{x}'(t)-\textbf{x}(t)||$. Stability can be quantified using the following measure 
\begin{equation}
    E(t)\equiv \frac{|DH(t)|}{|DH(0)|}.
\end{equation}
In Fig. 1 (a) we show how $E(t)$ of the brachistochrone equations \eqref{systembrachistochrone}  behave for the three different cases. The initial conditions have been chosen so that they realize a particular $\hat U_d$ for the integrable decompositions of Type I and Type II, and a general chaotic case, respectively. 
We observe that for a generic AB decomposition, even for SU(4) case, we observe Lyapunov exponents $\lambda_e \geq 1$ ie. $E^{ch}(t) \propto \exp(\lambda_e t)$.
where the "ch" superscript refers to the particular general chaotic decomposition. This indicates that the chaotic behaviour becomes relevant for the brachistochrone problem, since it kicks in for time $t<1$.

In order to exclude the role of exceptional perturbations, we sampled many different perturbations and we sampled $DH(t)$ over them.  Moreover, we tried numerous other AB decompositions. We confirm therefore the expected exponential divergence for a general AB decomposition.

\begin{figure}[h]
    \centering
    \includegraphics[width=8.66cm, height=5.2cm]{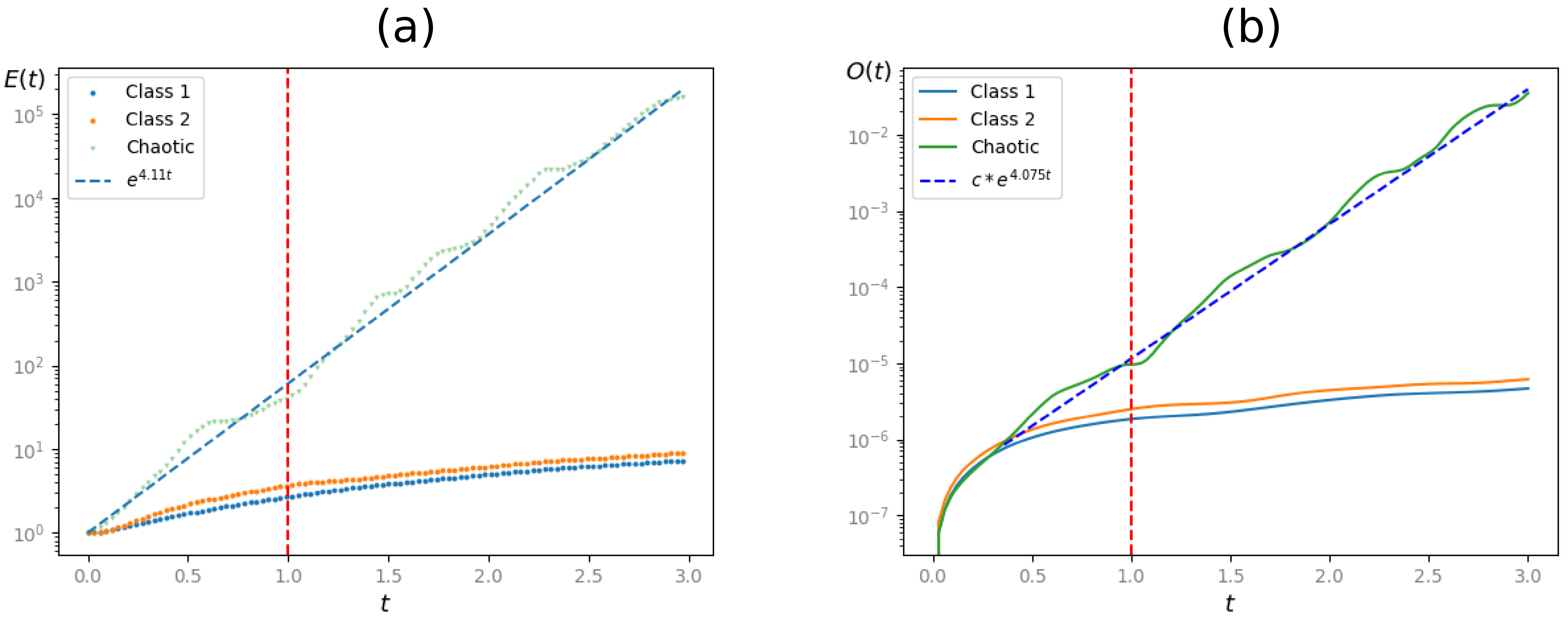}
    \caption{Stability of brachistochrone equations for $\mathfrak{su}(4)$. Panel (a): Given a particular $\hat U_d$ the respective $\textbf{x}(0)$ are found for the three protocols. The stability of those initial points are shown after averaging $E(t)$ for 2000 random small $\textbf{d}$. Panel (b): The time dependence of $O^i(t)$ for the same $\hat U_d$. The dashed lines fit the chaotic $E(t)$, providing estimations for the effective Lyapunov exponent.}
    \label{divt}
\end{figure}

It is instructive to see how the instability 
of the brachistochrone equations \eqref{systembrachistochrone} propagates to the generated unitary operator through \eqref{schrodinger}. 
To this end we look into the following measure 
of the divergence of to nearby trajectories
\begin{equation}
    O^i(t)=\langle||\text{Log}((\hat U_b^i)^{\dagger}(\textbf{x}_0,t) \hat U^i_b(\textbf{x}_0+\textbf{d},t))||_F\rangle_{\textbf{d}}.
\end{equation}
Here $U^i_b(\textbf{x}_0,t)$ is the generated unitary operator at time t for the initial conditions $\textbf{x}_0$ to the boundary value problem of \eqref{systembrachistochrone} for the i-th case (Integrable or Chaotic). By doing so, (see Panel (b) of Fig. 1), we can see when the non-linear behaviour appears for the chaotic case. For short enough times, the divergence of $\hat U_b(\textbf{x}_0+\textbf{d})$ from $\hat U_b(\textbf{x}_0)$ is linear for all three cases. For the chaotic one, we confirm this stops at time $t \approx \lambda_e ^{-1}$.  We confirm, therefore, that indeed the instability of the brachistochrone equation propagates to the generated unitary operator.

This comes as no surprise. If we assume
\begin{equation}
    \hat U^i_b(\textbf{x}_0+\textbf{d},t)=\hat U^i_b(\textbf{x}_0,t) \hat{u}(t),
\end{equation}
where the unitary operator $\hat u$ is the (conjugate transpose) of the argument in $\text{Log}$ of $O^i$. It obeys the Schr\"{o}dinger equation
\begin{equation}
    i \partial \hat u =  \big (\hat U^i_b(\textbf{x}_0,t) \big)^{\dagger} \hat \Delta\hat U^i_b(\textbf{x}_0,t) \hat u, \quad \hat u(0)=\mathbb{I},
\end{equation}
where $\hat \Delta=\hat{H}(t)'-\hat H(t)$ the difference of the two Hamiltonians for different initial conditions. This means that $\hat u$ is driven by the Hamiltonian (in the rotated frame) $\hat U^{\dagger}_b(\textbf{x}_0,t) \hat{\Delta}\hat U_b(\textbf{x}_0,t)$. Therefore one can explain why there is similar (albeit latent) Lyapunov exponent for the same equation.

Lyapunov instability is a fundamental manifestation of 
chaos in brachistochrone problems with non-integrable AB decompositions. However, a given AB decomposition 
is not generally characterised by a given 
largest Lyapunov exponent. Rather, the Lyapunov 
exponent is a function of the initial condition 
of the trajectory. To illustrate this point we 
investigate the statistics of the largest Lyapunov 
exponents in a chaotic system with a 
given AB decomposition. First, we notice that the brachistochrone equations \eqref{systembrachistochrone} are invariant under the transformation $\textbf{a}' \rightarrow \kappa \textbf{a}$, $\boldsymbol{\lambda}' \rightarrow \kappa \boldsymbol{\lambda}$, $t'\rightarrow {t}/{\kappa}$ where $\kappa \in \mathbb R^+$. This means that under this transformation the left hand side of
\begin{equation}
        \label{exponentialdiv}
         \frac{||\textbf{x}_1(t)-\textbf{x}_0(t)||}{||\textbf{x}_1(0)-\textbf{x}_0(0)||} \approx \exp( \gamma_0 t)
\end{equation}
remains invariant, while the right hand side becomes $\exp(\kappa \gamma_0 t)$. Therefore the transformed solution has a Lyapunov exponent $\kappa \gamma_0$. For this reason it is only meaningful to 
talk about the distribution of Lyapunov exponents 
on the basin of initial conditions defined by the 
constraint $||\textbf{x}_0||=1.$

 In Fig. \ref{Lyapunovprop} we show the distribution of Lyapunov exponents for random sampling of $\textbf{x}_0$'s with $||\textbf{x}_0||=1.$
For random sampling we assume the probability distribution invariant under the adjoint action of SU(n). 
\begin{figure}[h]
    \centering
    \includegraphics[width=8.06cm, height=7.3cm]{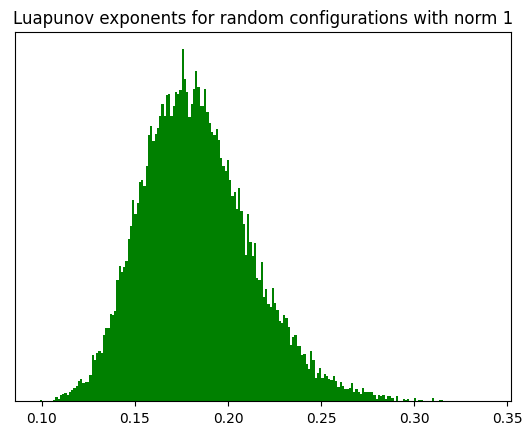}
    \caption{Properties of Lyapunov exponents for an SU(4) chaotic AB decomposition. The histogram of distribution of Lyapunov exponents for configurations $\textbf{x}_0$ with $||\textbf{x}_0||=1$. 25000 random configurations are sampled.} 
    \label{Lyapunovprop}
\end{figure}

We observed that for a chaotic protocol, close initial conditions diverge exponentially, while integrable protocols diverge polynomially. We expect, thus, between these two a qualitative difference regarding the closeness of two different unitary operators $\hat U_0, \hat U_1$ derived from close initial configurations $(\textbf{a}_0(0),\boldsymbol{\lambda}_0(0))$, $(\textbf{a}_1(0),\boldsymbol{\lambda}_1(0))$, respectively. Thus we introduce the measure:
\begin{equation}
    \label{measuref}
    \mathcal{F}^{i}(\hat U_0,\textbf{d})=||\text{Log}(( (\hat U_0)^\dagger \hat U^i_b(\textbf{x}_0+\textbf{d},1))||_F,
\end{equation}
where $\textbf{x}_0$ is the configuration that generates $\hat U_0$ through the brachistochrone equations for some particular protocol, and $\textbf{d}$ is a small perturbation. The index i refers to which protocol is used (an integrable or a chaotic one).
For a fixed $\textbf{x}_0$, we are going to constrain ourselves to $\textbf{d}$ of fixed (small) norm. Essentially $\textbf{x}_0$ determines the nature of the measure \eqref{measuref}. Thus, as we explained above, the norm of $\textbf{x}_0$ is the parameter that determines the closeness of the measure $\mathcal{F}^{i}$. When the norm of configuration $\textbf{x}_0$ is sufficiently small the statistical behavior of $log(\mathcal{F})$ between integrable and chaotic protocols for different $\textbf{d}$ (see Fig. \ref{divprop}(a)) will not be significant. On the other hand, when the norm of $\textbf{x}_0$ becomes bigger, there is an exponential separation and small deformations of configurations lead to bigger $\mathcal{F}$ (Fig. \ref{divprop}(b)). The difference in the behavior between the two cases is because for small Lyapunov exponents (small norm of $\textbf{x}_0$) the chaos doesn't play important role by the end of protocol (for times $t \approx 1$).

This observation hints at the impact of the increase 
in the rank of the unitary group and the number 
of constraints on the numerical hardness of 
computing the brachistochrone reaching an arbitrary 
unitary $\hat U_d$.
Generally, the higher the dimensionality of the Hilbert space, the greater the norm of the typical driving Hamiltonian. Furthermore, by introducing extra constraints one replaces "physical" degrees of freedom with Lagrange multipliers, which typically take bigger values. Both factors essentially increase the norm of the configurations $\textbf{x}_0$, which makes the brachistochrone path to a given $\hat U_d$ more unstable, as we established previously. We would like to stress, however, that no quantitative claim is made at this point and is an open question to be investigated.

\begin{figure}[h]
    \centering
    \includegraphics[width=8.66cm, height=5.7 cm]{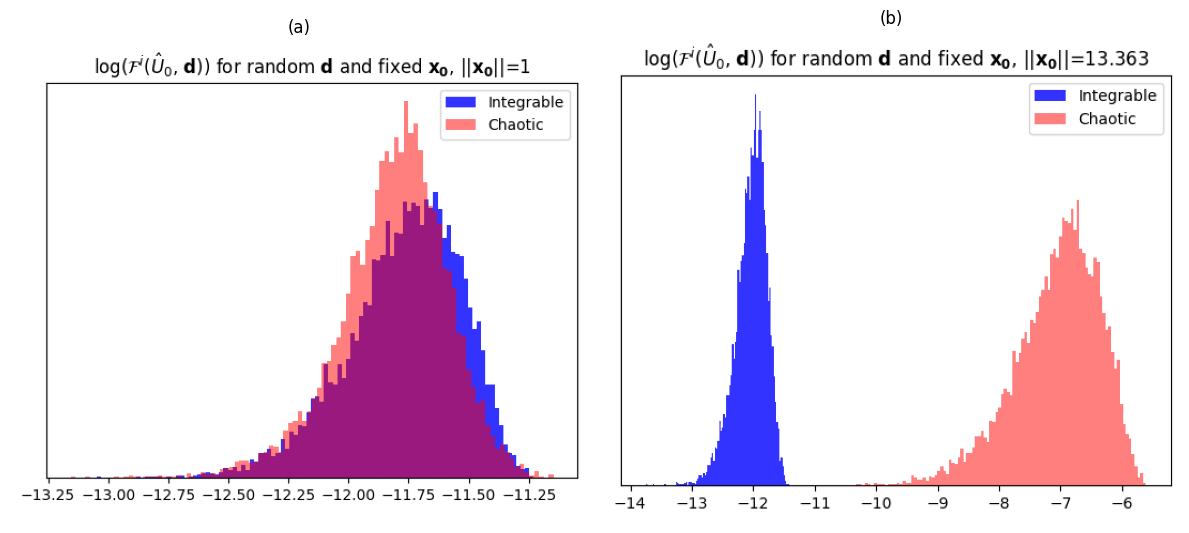}
    \caption{The statistics of $\log (\mathcal{F}^i(\hat U_0 ,\textbf{d}))$ for an integrable and a chaotic SU(4) AB decomposition and a given $\hat U_0$ (generated by $\textbf{x}_0$)  over random deformations $\textbf d$. Panel (a): A  configuration $\textbf{x}_0$ is used norm 1 is used. The behaviour of $\mathcal{F}$ between chaotic and integrable is qualitatively the same. Panel (b): A $\textbf{x}_0$ is used that solves the boundary value problem for some some random operator (i.e. not close enough to identity operator). We observe some exponential separation between the integrable and the chaotic one.} 
    \label{divprop}
\end{figure}

Instability with respect to small variations in the initial conditions makes the boundary value problem a numerically hard task. Essentially it turns out to be an optimization problem, where a high precision of the parameter landscape is required in order to converge to the desired unitary operator. To quantify this we introduce the cost function:
\begin{equation}
    C^{i}_{\hat{U}_{0}}(\textbf{x}):= || \text{Log} \big(\hat U_0^\dagger \hat U_{b,i}(\textbf{x},1) \big) ||_F,
\end{equation}

After fixing $\hat U_0$ and the AB decomposition, the minimization of $C$ requires some minimization routine. We employed certain pre-constructed ones for the  the chaotic and the integrable cases, respectively. Since  certain initial guesses may get stuck at some local minima, and not at the global one, we sampled over many initial points $\textbf{x}_0$ and study their statistics. From Fig. \ref{optimiedc} one can see that for the integrable case, there are many initial guesses that converge to the desired solution, with certain accuracy (of order of $10^{-6}$). On the other hand, in the chaotic case the optimization gets stuck to sub-optimal solutions and need extra resources to reach to the desired solution. So we deduce how immensively more difficult is to solve the boundary value problem for a generic decomposition, even for the SU(4) group, let alone higher ones.

\begin{figure}[h]
    \centering
    \includegraphics[width=8.66cm, height=6.7cm]{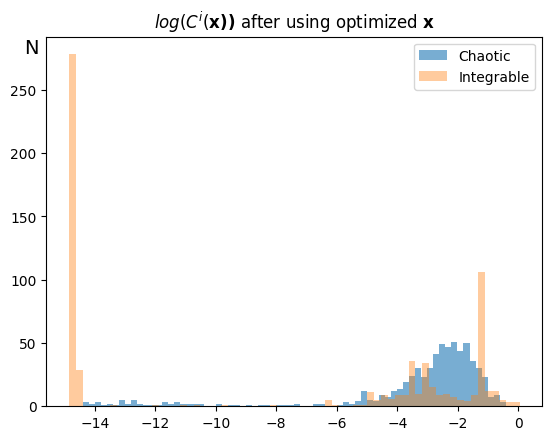}
    \caption{Histogram of the eventual values of $\log C ^i(\textbf{x})$ for an SU(4) chaotic and integrable AB decomposition. Some random initial values $\textbf{x}_0$ are used along with standard optimizing routines,  giving the final $\textbf{x}$.} 
    \label{optimiedc}
\end{figure}

`

\section{Conclusions}
We have investigated quantum brachistochrone equations describing the optimal realization of an SU(n) gate 
with the help of a protocol utilising a constrained set of driving Hamiltonians. We have 
found that for a certain class of driving Hamiltonians, 
the quantum brachistochrone equation can be viewed as 
a limiting case of a completely integrable system 
known as the generalised Euler-Arnold 
top. We give an explicit prescription for the
associated AB decomposition of the algebra of physical observables into 
the space of driving Hamiltonians and its orthogonal complement. 

In order to demonstrate the utility of the 
completely integrable AB decompositions, we 
compare the numerical stability of generic
brachistochrone equations with the 
completely integrable ones. In contrast to the 
integrable case, generic brachistochrone 
equations are found 
to exhibit exponential divergence 
of nearby trajectories characteristic of chaotic 
behaviour. We quantify such divergence by the 
Lyapunov exponents and investigate the statistical 
distribution and scaling properties of such 
exponents for small groups. 
We propose arguments as to why such exponential 
divergence poses an increasing numerical challenge for the solution of the boundary value problem as one increases the size of the unitary group or the number of constraints. 
This should motivate further investigation of
completely integrable protocols and their small 
non-integrable deformations.


Notwithstanding the intriguing link between integrable brachistochrone equations and the classical integrable 
models, many questions remain unanswered. 
For instance, it is unclear how far one can advance 
with the program of developing an explicit solution of the completely integrable brachistochrone equations, in particular with finding an explicit relationship between the initial conditions and the generated unitary operator at the end of the protocol. One systematic approach to 
this problem employs the Baker-Akhiezer functions and a factorisation method based on the matrix Riemann-Hilbert problem \cite{Semenov1981}. Furthermore, even if 
a formal algebraic solution to the boundary value 
problem is found, the actual computation of the initial 
conditions may still present a certain challenge as is 
evidenced by Eq. \eqref{multiplyunitary}.
It is also important to note that the completely integrable cases of AB decompositions identified here may not cover all integrable brachistochrone problems, making further investigation into this area an interesting and relevant problem in its own right


{\it Acknowledgements} This publication is part of the project Adiabatic Protocols in Extended Quantum Systems, Project No 680-91-130, which is funded by the Dutch Research Council (NWO).

\nocite{*}

\bibliography{main.bib}

\widetext

\appendix
\begin{section}{APPENDIX: Acquiring the brachistochrone equations}
Let's consider the action with the Lagrange multipliers $\{\lambda_i\}$:
\begin{equation}
\label{actionunit}
    S_0=\int_0^{T}dt \Big(\frac{1}{2}Tr(\partial_t \hat{U} \partial_t \hat{U}^{\dagger})+i\sum_{\kappa} \lambda_{\kappa}Tr(B_k \partial_t \hat{U} \hat{U}^{\dagger})\Big)
\end{equation}
To facilitate our computations we adopt the index notation. This way the matrix multiplication and Tr operation get simpler. The integrand, then, is written:
\begin{equation}
    s=\frac{1}{2}(\partial \hat{U})_{ab} (\partial \hat{U}^{\dagger})_{ba}+i\sum_{\kappa}\lambda_{\kappa} (B_{\kappa})_{ab}(\partial\hat{U})_{bc}(U^{\dagger})_{ca}
\end{equation}
We are going to vary the "field" U to minimize the action s. To keep it short we separate the action in two part, $s_{1,2}$, respectively:
\begin{equation}
    \delta s_1= \frac{1}{2} (\partial \delta \hat{U})_{ab}(\partial\hat{U}^{\dagger})_{ba}+\frac{1}{2} (\partial\hat{U})_{ab} (\partial \delta\hat{U}^{\dagger})_{ba}
\end{equation}
Up to boundary terms, $\delta s_1$ can be written:
\begin{equation}
    \delta s_1=-\frac{1}{2}\Bigg((\partial^2 \hat{U}^{\dagger})_{ba}(\delta \hat{U})_{ab}+(\partial^2 \hat{U})_{ab}(\delta\hat{U}^{\dagger})_{ba}\Bigg)
\end{equation}
Meanwhile the variation of $s_2$:
\begin{equation}
    \delta s_2=i\sum_{\kappa}\Big( \lambda_{\kappa} (B_{\kappa})_{ab}(\partial\delta\hat{U})_{bc} (\hat{U}^{\dagger})_{ca} + \lambda_{\kappa} (\hat B_{\kappa})_{ab}(\partial\hat{U})_{bc} (\delta \hat{U}^{\dagger})_{ca}\Big)
\end{equation}
Again up to boundary terms:
\begin{equation}
    \delta s_2 = i\sum_{\kappa} \Bigg(-\dot \lambda_{\kappa} (\hat{B}_{\kappa})_{ab} (\delta \hat{U})_{bc} (\hat{U}^{\dagger})_{ca}-\lambda_{\kappa}(\hat B_{\kappa})_{ab}(\delta\hat{U})_{bc}(\partial\hat{U}^{\dagger})_{ca}+\lambda_{\kappa} (\hat B_{\kappa})_{ab}(\partial\hat{U})_{bc} (\delta \hat{U}^{\dagger})_{ca}\Bigg)
\end{equation}
We want to collect all the variation wrt $\delta U$. So we need to take into consideration $\hat{U}\hat{U}^{\dagger}=\mathbb{I} \Rightarrow \delta \hat{U} \hat{U}^{\dagger}+\hat{U}\delta \hat{U}^{\dagger}=0 \Rightarrow \delta \hat{U}^{\dagger}= -\hat{U}^{\dagger} \delta \hat{U} \hat{U}^{\dagger}$. Moreover using the Tr properties, we can re-write $\delta s_1 =\delta \hat{U}_{ab}T^{1}_{ba}, \delta s_2 =\delta \hat{U}_{ab}T^{2}_{ba}$:
\begin{gather}
    T^1=-\frac{1}{2}(\partial^2 \hat{U}^{\dagger}-\hat{U}^{\dagger}\partial^2\hat{U} \hat{U}^{\dagger}) \\
    T^2=-i \sum_{\kappa} \lambda_{\kappa} \partial \hat{U}^{\dagger} \hat B_{\kappa}-i\sum_{\kappa} \dot \lambda_{\kappa} \hat{U}^{\dagger} \hat{B}_{\kappa}-i\sum_{\kappa} \lambda_{\kappa}\hat{U}^{\dagger} \hat{B}_{\kappa} \partial\hat{U} \hat{U}^{\dagger}
\end{gather}
So essentially the "equation of motion" is $T^1+T^2=0$. To simplify it further, we use  $\partial \hat{U}^{\dagger} =-\hat{U}^{\dagger} \partial \hat{U} \hat{U}^{\dagger}$. So we'll have:
\begin{equation}
    T^1=\hat{U}^{\dagger} \partial^2 \hat{U} \hat{U}^{\dagger}+\frac{1}{2}( \partial \hat{U}^{\dagger} \partial \hat{U} \hat{U}^{\dagger}+\hat{U}^{\dagger} \partial \hat{U} \partial U^{\dagger})
\end{equation}
Multiplied with $\hat U$ $T^1$ becomes:
\begin{equation}
    UT^1= \partial^2 \hat{U} \hat{U}^{\dagger} +\frac{1}{2}(\hat{U}\partial U^{\dagger} \partial \hat{U} \hat{U}^{\dagger}+\partial \hat{U} \partial U^{\dagger})
\end{equation}
The terms in the parenthesis are equal (use twice the property $\partial \hat U \hat{U}^{\dagger}=-\hat U \partial\hat{U}^{\dagger} $. Actually, the whole term $i \hat{U} T^1=\partial_t \hat{H}$. Similar manipulation of $\hat U T^2$ gives:
\begin{equation}
    \hat{U} T^2=\sum_{\kappa}\Big( \lambda_{\kappa} \hat{H} \hat{B}_{\kappa}-i \dot{\lambda}_{\kappa} \hat{B}_{\kappa}-\lambda_{\kappa} \hat{B}_{\kappa} \hat{H})=\sum_{\kappa}\Big( \lambda_{\kappa}[\hat{H},\hat{B}_{\kappa}]-i\dot{\lambda}_{\kappa} \hat{B}_{\kappa}\Big)
\end{equation}
As a whole $i\hat{U}(T^1+T^2)=0$ gives:
\begin{equation}
    \partial_t \hat{H} +\sum_{\kappa}\Big( i\lambda_{\kappa}[\hat{H},\hat{B}_{\kappa}]+\dot{\lambda}_{\kappa} \hat{B}_{\kappa}\Big)=0
\end{equation}
This is the known quantum brachistochrone equation. Also the second term we introduced implies:
\begin{equation}
    Tr(\hat H \hat{B}_k)=0, \quad \forall k 
\end{equation}

\textbf{Note}: The variation of action $S_0$ is equivalent to the variation of the action $S_1$:
\begin{equation}
    S_1=\int dt\Bigg(\sqrt{Tr(\partial U \partial U^{\dagger})}+\sum_i \lambda_{i}Tr(\partial \hat U \hat U^{\dagger})\Bigg)
\end{equation}
The only thing that changes is that $S_0(\lambda_{_i})\rightarrow S_1(\lambda_{i}/\lVert H \rVert^2_F)$; but as we show in a next section the Frobenius norm of H remains constant along a trajectory. This means that we should need to rescale the Lagrange multipliers.
\end{section}

\begin{section}{An example of a linearizable AB decomposition}
We consider the $\mathfrak{su}(3)$ algebra and its pseudo-Cartan decomposition $\mathfrak{g}=\mathfrak{l}+\mathfrak{p}$ such that 
$\hat s\in \mathfrak p$ and $\hat l\in \mathfrak l$ are parametrised 
as follows
\begin{equation}
    \hat s=\begin{pmatrix}
        0 & 0 & z_1 \\
        0 & 0 & z_2 \\
        z_1^* & z_2^* & 0
    \end{pmatrix}, \quad  \hat l=\begin{pmatrix}
        r_1 & \psi_1 & 0 \\
        \psi_1^* & r_2 & 0 \\
        0 & 0 & -(r_1+r_2)
    \end{pmatrix}
\end{equation}
Note that the subgroup $\mathfrak l \cong \mathfrak{su}(2) \cross \mathfrak u(1)$.Now let's fix the element $\hat a$. Let that be (arbitrarily):
\begin{equation}
   \hat a=\begin{pmatrix}
        0 & 0 & a_1 \\
        0 & 0 & 0 \\
        a_1 & 0 & 0
    \end{pmatrix},
\end{equation}
where $a_1\in \mathbb{R}$. From that we infer that the subset $\mathfrak{l}_a$ consists of the single (linearly independent) element:
\begin{equation}
   \hat X=\frac{1}{\sqrt{3}}\text{diag}(1,-2,1)
\end{equation}
Therefore $L_p=\sum_{i} m_i Y_i$, where $L_p\in \mathfrak{l}^{\perp}_a$ and $\mathfrak{l}_a^{\perp}$:
\begin{equation}
    \mathfrak{l}_a^{\perp} :=\text{span}\Bigg\{ 
   \hat Y_1=\begin{pmatrix}
        0 & 1 & 0 \\
        1 & 0 & 0 \\
        0 & 0 & 0
    \end{pmatrix}, \hat Y_2=\begin{pmatrix}
        0 & -i & 0 \\
        i & 0 & 0 \\
        0 & 0 & 0
    \end{pmatrix}, \hat Y_3=\text{diag}(1,0,-1)
    \Bigg\}
\end{equation}
Moreover, the generic invariant polynomial can be generated by:
\begin{equation}
    \phi= \sum_{i} c_i Tr\Big( (\hat s+\hat l)^i\Big),
\end{equation}
for arbitrary $\{c_i\}$. From that on can define $b=d\phi(a)$ (i.e. the gradient of $\phi$ at $a$). If one includes at least powers $2,3$ basically what they get is:
\begin{equation}
    b= \kappa_1 a +\kappa_2 X,
\end{equation}
$\kappa_i$'s are arbitrary (and dependent on all $c_i$). From this, we can get $\omega$. In particular:
\begin{equation}
    \omega(l) =
    \begin{cases}
      \mu X & \text{if $l \in \mathfrak{l}_a$} \\
      ad_b (ad_a)^{-1} l & \text{if $l\in \mathfrak{l}^{\perp}_a$}
    \end{cases} 
\end{equation}
again $\mu$ is related to $c$ (but independent from $\kappa_i$'s). Let's see that what happens first with $(ad_a)^{-1}$. We want some operator Z that $i[a,Z]=l$ (physics convention with multiplication with i). So Z for each of the 3 cases is:
\begin{equation}
    Z=\begin{cases}
        \begin{pmatrix}
            0 & 0 & 0 \\
            0 & 0 & \frac{i}{a_1} \\
            0 & -\frac{i}{a_1} & 0
        \end{pmatrix}, & \text{if $l=Y_1$} \\
        \begin{pmatrix}
            0 & 0 & 0 \\
            0 & 0 & -\frac{1}{a_1} \\
            0 & -\frac{1}{a_1} & 0
        \end{pmatrix}, & \text{if $l=Y_2$} \\
        \begin{pmatrix}
            0 & 0 & \frac{i}{2a_1} \\
            0 & 0 & 0 \\
            \frac{-i}{2a_1} & 0 & 0
        \end{pmatrix}, & \text{if $l=Y_3$} 
    \end{cases}
\end{equation}
And consequently the last step is to get simply $i[b,Z]$ for the three different cases;

\begin{equation}
    \omega(l) =
    \begin{cases}
      \mu X & \text{if $l \in \mathfrak{l}_a$} \\
      \frac{\kappa_1}{a_1} Y_1  & \text{if $l=Y_1$} \\
      \frac{\kappa_1}{a_1} Y_2 & \text{if $l=Y_2$} \\
      \frac{\kappa_1}{a_1} Y_3 & \text{if $l=Y_3$}
    \end{cases} 
\end{equation}
Note that to reach to this result we had to set $\kappa_2=0$, since with non-zero $\kappa_2$ it would generate components of $\omega$ that don't belong in $\mathfrak{l}$. With this is mind, we are able to write down the final differential equations. By using the linearity of $\omega(l)$:
\begin{equation}
\label{eulerlileeqs}
    \dot l=[l,\omega(l)]+[s,b], \quad \dot s=[s,\omega],
\end{equation}
If we take the $a_1\rightarrow0$ limit (that is the matrix $a\rightarrow 0$). The system of equations becomes the brachistochrone equation:
\begin{equation}
    \dot l +\dot s=[l+s,\omega]
\end{equation}
Note that the limit is legitimate since $\omega$ remains finite in that limit. With this, we can set $\mu=0, \kappa_1=\alpha_1$, in order to make $\omega$ a projector onto the $\mathfrak l^{\perp}$. This way we create a new integrable protocol. The system of equations one has to solve is:
\begin{equation}
    \begin{pmatrix}
        \dot a_1 \\
        \dot a_2 \\
        \dot a_3 \\
        \dot a_4 \\
        \dot l \\
        \dot m_1 \\
        \dot m_2 \\
        \dot m_3
    \end{pmatrix}
    =
    \begin{pmatrix}
        -a_2 m_2 - a_3 (m_1 + 2 m_3)\\
        a_1 m_2 - a_3 (m_1 + m_3) \\
        a_2 m_1 - a_3 m_2 + 2 a_1 m_3\\
        a_1 m_1 + a_3 m_2 + a_2 m_3\\
        0\\
        \sqrt{3} l m_2\\
        -\sqrt{3} l m_1\\
        0
    \end{pmatrix}, \quad \hat{H}(t)=\begin{pmatrix}
        \frac{l}{\sqrt{3}} & 0 & a_1- i a_3 \\
        0 & -\frac{2l}{\sqrt{3}} & a_2-i a_4 \\
        a_1+ i a_3 & a_2+i a_4 &\frac{l}{\sqrt{3}}
    \end{pmatrix}
\end{equation}
Note that the system is linear. In particular, $m_1,m_2$ can be solved exactly and are some linear combination of $\sin(\sqrt{3}lt), \cos(\sqrt{3}lt)$. Also since $m_3$ is time-independent, the system of $a_i$ becomes a driven linear system. The explicit solution is given:
\begin{equation}
\begin{pmatrix}
         a_1(t) \\
         a_2(t) \\
         a_3(t) \\
         a_4(t) \\
\end{pmatrix}
=\mathcal{T} \exp( \int_{0}^{t}
 \hat M(t') dt') \begin{pmatrix}
         a_1(0) \\
         a_2(0) \\
         a_3(0) \\
         a_4(0) \\
\end{pmatrix}, \quad \hat M(t)=\begin{pmatrix}
    0 & -m_2(t) & -2 m_3(t) & -m_1(t) \\
    m_2(t) & 0 & -m_1(t) & -m_3(t) \\
    2 m_3(t) & m_1(t) & 0 & -m_2(t) \\
    m_1(t) & m_3(t) & m_2(t) & 0
\end{pmatrix}
\end{equation}
Since the matrix $\hat M$ is antisymmetric, we know that the matrix remains finite and there aren't exponential divergences.
\end{section}

\end{document}